\documentclass{article}
\usepackage{spconf}
\usepackage[utf8]{inputenc}
\usepackage[T1]{fontenc}

\usepackage{cite}

\usepackage{amsmath,amssymb,amsthm,bm,xfrac}
\DeclareMathSymbol{\shortminus}{\mathbin}{AMSa}{"39}
\newtheorem{theorem}{Theorem}

\usepackage{mathtools}

\usepackage{graphicx,subcaption}
\usepackage{booktabs,multirow}

\usepackage{pgfplots}
\pgfplotsset{compat=1.17}

\usepackage{xcolor}
\definecolor{comb_lapl}{rgb}{0,0.4470,0.7410}
\definecolor{rw_lapl}{rgb}{0.8500,0.3250,0.0980}
\definecolor{voro_lapl}{rgb}{0.9290,0.6940,0.1250}

\usepackage[inline]{enumitem}
\newlist{inlinelist}{enumerate*}{1}
\setlist*[inlinelist,1]{label=\roman*),itemjoin={{, }},itemjoin*={{, and }}}

\usepackage{hyperref}
\def\equationautorefname~#1\null{(#1)\null}

\usepackage{xspace}

\def\qmin{\ensuremath{q_{\mathrm{min}}}}

\DeclareMathOperator{\diag}{diag}
\DeclareMathOperator{\trace}{tr}
\DeclareMathOperator{\logdet}{logdet}

\gdef\eusipco{0}

\def\resp{\textit{resp.}\xspace}

\title{Joint Graph and Vertex Importance Learning}

\name{Benjamin Girault\textsuperscript{*}, Eduardo Pavez\textsuperscript{$\dagger$}, Antonio Ortega\textsuperscript{$\dagger$}}

\address{\textsuperscript{*} Université de Rennes, ENSAI, CNRS, CREST-UMR 9194, Rennes, France\\\textsuperscript{$\dagger$} Signal and Image Processing Institute, University of Southern California, Los Angeles, USA}

\begin{document}
\setlength{\abovedisplayskip}{4pt}
\setlength{\belowdisplayskip}{4pt}
\maketitle
\begin{abstract}
In this paper, we explore the topic of graph learning from the perspective of the Irregularity-Aware Graph Fourier Transform, with the goal of learning the graph signal space inner product to better model data.
We propose a novel method to learn a graph with smaller edge weight upper bounds compared to combinatorial Laplacian approaches.
Experimentally, our approach yields much sparser graphs compared to a combinatorial Laplacian approach, with a more interpretable model.
\end{abstract}

\section{Introduction}
\label{sec:intro}

The field of graph signal processing proposes a toolbox to analyze, process and transform data supported by arbitrary discrete structures \cite{Ortega.PROCIEEE,Ortega.BOOK.2022}.
Examples include data collected by sensor networks, activities in a human neuronal network, or image processing \cite{Ortega.PROCIEEE}.
However, for some applications, we have access to the data, but the graph structure is either missing, noisy, or incomplete.
In this case, it is important to learn a graph that provides a model for the data that can be leveraged by the graph signal processing toolbox to process or analyze the data.
To that end, graph learning is typically formulated as an optimization problem, where the goal is to obtain an algebraic representation of the graph that best captures the variability of the data.
Popular data fidelity objectives include graph signal smoothness and stationarity \cite{dong2019learning,Mateos.SPMAG.2019}.

Following \cite{Egilmez.JSTSP.2017}, graph learning can be formulated as an inverse covariance estimation problem with graph Laplacian constraints.
In this context, the most common type of Laplacian is the \emph{Combinatorial Graph Laplacian} (CGL).
The CGL variation $\Delta(\mathbf{x})=\mathbf{x}^*\mathbf{L}\mathbf{x}=\frac{1}{2}\sum_{ij}w_{ij}|x_i-x_j|^2$ exhibits the desirable property of constant signals having zero variation, since only variations along edges are used \cite{Ortega.PROCIEEE}.
In addition, the \emph{Graph Fourier Transform} (GFT) decomposes a signal on an orthonormal basis that minimizes this variation \cite{Girault.TSP.2018}.
Learning a CGL leads then to a consistent spectral interpretation given by the \emph{graph Power Spectrum Density} (gPSD).
Ultimately, learning a CGL and having access to such a gPSD allows for specific filter designs such as Wiener filters \cite{Girault.ICASSP.2014} or ARMA filters \cite{Marques.TSP.2017}.
To account for the difficulty associated with singular CGL matrices in inverse covariance estimation, the objective function is oftentimes modified \cite{Egilmez.JSTSP.2017,Pavez.ASILOMAR.2019,Pavez.AISTATS.2022,ying2020nonconvex,ying2021minimax}.
However, such an approach produces dense graphs, even if variables are weakly correlated (see \autoref{sec:experiments} and \cite{ying2020nonconvex}) because the modified objective function encourages well connected graphs \cite{Pavez.ASILOMAR.2019}.
This issue can be solved by incorporating non-convex sparse regularization \cite{ying2020nonconvex,koyakumaru2023learning} at the expense of a more complex graph learning algorithm.

This strict CGL approach can be relaxed by allowing for the estimated inverse covariance to be a CGL plus a diagonal matrix  \cite{Egilmez.JSTSP.2017}.
The resulting matrices are the \emph{Generalized Graph Laplacian} (GGL) and the \emph{Diagonally Dominant Graph Laplacian} (DDGL).
Many graph learning algorithms for GGL matrices have been proposed \cite{Egilmez.JSTSP.2017,Pavez.ICASSP.2016,slawski2015estimation,ying2022adaptive}, which have been shown not to produce the spurious connections that often arise in CGL approaches \cite{Pavez.TSP.2018}.
However, the graph variation loses its classical interpretations with constant signals  exhibiting non-zero variations since an additional term accounting for the signal weighted magnitude on each vertex.

In this work we propose a new formulation that allows us to learn CGLs without spurious connections while preserving interpretability, in contrast with the relaxed Laplacians above for the DDGL case.
Our approach to the inverse covariance problem uses  our recently introduced generalization of the \emph{Graph Fourier Transform} (GFT) to arbitrary Hilbert spaces of graph signals: the \emph{Irregularity Aware Graph Fourier Transform} \cite{Girault.TSP.2018}.
In this generalization, the space of graph signals is equipped with an inner product other than the standard dot product, thus adapting the notion of \emph{orthogonality} and \emph{norm} of these graph signals to the application.
This additional parameter to the GFT shows great promise in areas such as vertex sampling \cite{Girault.ICASSP.2020}, image \cite{Lu.ICIP.2020} or point cloud processing \cite{Pavez.ICIP.2020}, and filter design \cite{Pavez.ARXIV.2022}.
In the context of graph learning, we propose to \textit{jointly learn a CGL and an inner product}, which corresponds to a diagonal matrix whose diagonal terms reflect the relative importances of the vertices \cite{Girault.TSP.2018}.

Our contributions are threefold:
\begin{inlinelist}
    \item we formulate a joint CGL and inner product learning problem and show that it can be reduced to learning a DDGL (\autoref{sec:proposed:model}), allowing for the CGL and the vertex importance weights to be learned using any DDGL algorithm \cite{Egilmez.JSTSP.2017,ying2022adaptive}
    \item we propose an efficient and scalable coordinate minimization algorithm similar to \cite{Pavez.ASILOMAR.2019} for the CGL that updates one edge weight or vertex importance per iteration (\autoref{sec:proposed:method})
    \item we prove that the proposed CGL and inner product solution is sparse and obtain a sharp upper bound for the non zero edge weights (\autoref{sec:proposed:edge_prop}).
\end{inlinelist}

Our experiments with sampled intrinsic stationary continuous signals highlight the key benefits of our approach, including spatial consistency and a substantial increase of sparsity compared to learning a CGL (\autoref{sec:experiments}).

\section{Background}

\subsection{Graph Signal Processing}

A graph $\mathcal{G}=(\mathcal{V},\mathcal{E},w)$ is defined by a set of vertices $\mathcal{V}$, a set of edges $\mathcal{E}\subseteq\mathcal{V}\times\mathcal{V}$ and an edge weight function $w:\mathcal{E}\rightarrow\mathbb{R}_+$.
We denote by $N$ (\textit{resp.} $M$) the number of vertices (\textit{resp.} edges).
In this paper we are interested in undirected graphs where for any edge $e=(i,j)\in\mathcal{E}$, its opposite $(j,i)$ is also an edge in $\mathcal{E}$ with identical weight $w(i,j)=w(j,i)$.

Algebraic representations of such graphs include the classical adjacency or weight matrix $\mathbf{A}$ such that $\mathbf{A}_{ij}=w(i,j)$ if $ij\in\mathcal{E}$, 0 otherwise.
The degree of a vertex is defined as the sum $d_i=\sum_{ij\in\mathcal{E}} w(i,j)$ of its incident edges, and these degrees are collected into the diagonal degree matrix $\mathbf{D}=\diag(d_1,\dots,d_N)$.
This allows to define the combinatorial Laplacian matrix $\mathbf{L}=\mathbf{D}-\mathbf{A}$ of the graph.

Using the unweighted incidence matrix $\mathbf{B}=[\mathbf{b}_1 \cdots \mathbf{b}_M]$ such that for any edge $e=(i,j)\in\mathcal{E}$, such that $i<j$, then $\mathbf{B}_{i,e}=1$ and $\mathbf{B}_{j,e}=-1$, and $\mathbf{B}$ is zero elsewhere, we have $\mathbf{L}=\mathbf{B}\mathbf{W}\mathbf{B}^{\smash{T}}$, where $\mathbf{W}=\diag(w_1,\dots,w_M)$ is the diagonal matrix collecting edge weights.

A graph signal $x:\mathcal{V}\rightarrow\mathbb{C}$ is a function mapping vertices to scalar values.
We assumed an indexing of the vertices with integers $\{1,\dots,N\}$ (and of the edges with $\{1,\dots,M\}$) to define the algebraic representations above.
A graph signal $\mathbf{x}$ is then represented by a column vector $[x_1,\dots,x_N]^T\in\mathbb{C}^N$.

\subsection{Irregularity Aware Graph Fourier Transform}

The Irregularity Aware Graph Fourier Transform is a parametric generalisation of the orthonormal Graph Fourier Transforms (GFT) using two parameters~\cite{Girault.TSP.2018}: the graph variation operator $\Delta:\mathbb{C}^N\rightarrow\mathbb{R}_+$ mapping any graph signal $\mathbf{x}$ to its non-negative variation $\Delta(\mathbf{x})\geq 0$, and an inner product $\langle .,.\rangle_\mathbf{Q}:\mathbb{C}^N\times\mathbb{C}^N\rightarrow\mathbb{R}$ with Hermitian positive definite matrix $\mathbf{Q}$, such that $\langle \mathbf{x},\mathbf{y}\rangle_\mathbf{Q}=\mathbf{y}^*\mathbf{Q}\mathbf{x}$.
We denote this as the  $(\Delta,\mathbf{Q})$-GFT.
With $\Delta(\mathbf{x})=\mathbf{x}^*\mathbf{M}\mathbf{x}$ and $\mathbf{M}$ a Hermitian semi-definite positive matrix, its graph Fourier modes $\{u_l\}_l$ verify $\mathbf{M}\mathbf{u_l}=\lambda_l\mathbf{Q}\mathbf{u_l}$.
The graph Fourier basis $\{u_l\}_l$ is then orthonormal with respect to the $\mathbf{Q}$-inner product.
Collecting all the eigenvalues in the diagonal matrix $\mathbf{\Lambda}=\diag(\lambda_0,\dots,\lambda_{N-1})$, and graph Fourier modes in $\mathbf{U}=[u_0 \dots u_{N-1}]$ (and we have $\mathbf{U}^*\mathbf{Q}\mathbf{U}=\mathbf{I}$).
We obtain the inverse and forward $(\Delta,\mathbf{Q})$-GFT of a graph signal $\mathbf{x}$ with
$\mathbf{x} = \mathbf{F}^{-1}\mathbf{\widehat{x}} = \mathbf{U}\mathbf{\widehat{x}}$ and $\mathbf{\widehat{x}} = \mathbf{F}\mathbf{x} = \mathbf{U}^*\mathbf{Q}\mathbf{x}$.

In this paper, we use the combinatorial Laplacian variation $\Delta(\mathbf{x})=\mathbf{x}^*\mathbf{L}\mathbf{x}$ ($\mathbf{M}=\mathbf{L}$), and further assume that the inner product matrix is diagonal $\mathbf{Q}=\diag(q_1,\dots,q_N)$.
This effectively adds vertex \emph{importances} to the model alongside edge weights.
Non-diagonal inner product matrices $\mathbf{Q}$ in the context of graph learning will be studied in a future communication.

\subsection{Graph Wide Sense Stationarity (gWSS)}

Our goal is to obtain a stochastic model for the data at hand.
To that end, we use as class of models the framework of Graph Wide Sense Stationarity (gWSS) we previously introduced and its spectral characterization \cite{Girault.EUSIPCO.2015}: A stochastic graph signal is \emph{gWSS} if and only if its mean is a DC component and its spectral components are uncorrelated.
We denote by $\bm\Sigma$ the covariance matrix of the graph signal, $\bm\Gamma$ its spectral covariance matrix, and $\bm\gamma$ its diagonal, called the graph Power Spectrum Density (gPSD) \cite{Girault.EUSIPCO.2015}.
In other words, $\bm\gamma_l$ is the variance of the $l^{\text{th}}$ spectral component $\mathbf{\widehat{x}}_l$ of the signal $\mathbf{x}$.

\subsection{Graph Learning through Coordinate Minimization}
\label{sec:asilomar}

Assume that the graph signals are realizations of a zero-mean Gaussian gWSS with gPSD $\gamma(0)=0$ and $\gamma(\lambda)=1/\lambda$ when $\lambda>0$.
Using the $(\mathbf{L},\mathbf{I})$-GFT, this translates into the covariance matrix $\bm\Sigma=\mathbf{L}^\dagger$.
Under this setting we obtain the maximum likelihood estimator proposed by \cite{Egilmez.JSTSP.2017} (see \autoref{sec:intro}), which minimizes the following cost function:
\begin{equation}\label{eq:asilomar_cost}
  F(\mathbf{L}) =
    -\logdet\left(\mathbf{L}+\sfrac{1}{N}\mathbf{J}\right)
    +\trace\left(\mathbf{L}\mathbf{S}\right)
  \text{,}
\end{equation}
where $\mathbf{J}$ is the all-one matrix (necessary, since a valid $\mathbf{L}$ is always singular), and $\mathbf{S}=\frac{1}{K}\sum_k\mathbf{x^{(k)}}\mathbf{x^{(k)}}^*$ is the empirical covariance matrix.
The coordinate minimization approach of \cite{Pavez.ASILOMAR.2019} minimizes $F$ iteratively for all edge weights.
The optimal update (fixing all other weights) $\delta_e$ to edge weight $w_e$ is:
\begin{equation}\label{eq:update_asilomar}
  \delta_e=\max(-w_e, \sfrac{1}{h_e}-\sfrac{1}{r_e})
  \text{,}
\end{equation}
with the edge cost $h_e=\mathbf{b_e}^{\smash{T}}\mathbf{S}\mathbf{b_e}$, and the effective resistance $r_e=\mathbf{b_e}^T(\mathbf{L}+\sfrac{1}{N}\mathbf{J})^{-1}\mathbf{b_e}$ for an edge $e$.
The graph learning algorithm updates iteratively all edge weights using \autoref{eq:update_asilomar}, and stops when the update to $F$ after all the edge weight updates is below a stopping threshold.
The optimal graph weights are upper bounded by $w_e\leq\sfrac{1}{h_e}$ \cite{Pavez.ASILOMAR.2019}.

\section{Proposed Graph Learning Approach}
\label{sec:proposed}

\subsection{Data Model}
\label{sec:proposed:model}

We assume that the data we have is i.i.d. multivariate Gaussian.
Our goal is to model the data with a graph given by its combinatorial graph Laplacian $\mathbf{L}$ and the graph signal inner product matrix $\mathbf{Q}=\diag(\mathbf{q})$.
To that end, we learn this graph such that the data is gWSS on the graph, with gPSD $\gamma(\lambda)=(1+\lambda)^{-1}$ when using the $(\mathbf{L},\mathbf{Q})$-GFT.
Such a gPSD allows for a continuous gPSD around frequency 0 compared to the classical $\gamma(\lambda)=1/\lambda$.

\begin{theorem}\label{thm:proposed_covariance}
    The covariance matrix of a stochastic graph signal with gPSD $\gamma(\lambda)=(1+\lambda)^{-1}$ using the $(\mathbf{L},\mathbf{Q})$-GFT is:
    \begin{equation}\label{eq:proposed_covariance}
        \bm\Sigma = \left[\mathbf{Q}+\mathbf{L}\right]^{-1}
        \text{.}
    \end{equation}
\end{theorem}

\begin{proof}
    We first observe that the spectral covariance matrix verifies $\mathbf{\Gamma}=[\mathbf{I}+\mathbf{\Lambda}]^{-1}$.
    Using the $(\mathbf{L},\mathbf{Q})$-GFT, we obtain $\mathbf{\Sigma}=\mathbb{E}[\mathbf{x}\mathbf{x}^*]=\mathbf{F}^{-1}\mathbf{\Gamma}\mathbf{F}\mathbf{Q}^{-1}=[\mathbf{I}+\mathbf{Q}^{-1}\mathbf{L}]^{-1}\mathbf{Q}^{-1}=[\mathbf{Q}+\mathbf{L}]^{-1}$.
    Given that $\mathbf{Q}$ is Hermitian positive definite and $\mathbf{L}$ is Hermitian positive semi-definite, then $\mathbf{Q}+\mathbf{L}$ is definite, hence invertible and $\bm{\Sigma}$ is well-defined.
\end{proof}

Using maximum likelihood, the Gaussian assumption and \autoref{thm:proposed_covariance} lead to following cost function to minimize:
\begin{equation}\label{eq:proposed_loglikelihood}
  F(\mathbf{L}, \mathbf{Q})=
    -\logdet(\mathbf{Q}+\mathbf{L})
    +\trace((\mathbf{Q}+\mathbf{L})\mathbf{S})
  \text{,}
\end{equation}
with $\mathbf{S}$ the empirical covariance matrix of the data.
Interestingly, the matrix $\mathbf{Q}+\mathbf{L}$ in \autoref{eq:proposed_loglikelihood} also corresponds to the \emph{Diagonally Dominant Graph Laplacian} proposed in \cite{Egilmez.JSTSP.2017}, but with strict dominance of the diagonal.
However, our spectral interpretation to this model using the $(\mathbf{L},\mathbf{Q})$-GFT is different.

Note also that the cost function \eqref{eq:proposed_loglikelihood} is actually a generalization of \cite{lake2010discovering} for any positive definite matrix $\mathbf{Q}$ instead of only scaled versions of the identity matrix $\frac{1}{\sigma^2}\mathbf{I}$.
The additional $\ell_1$ penalty term of \cite{lake2010discovering} and the difference between $\mathbf{Q}$ and $\frac{1}{\sigma^2}\mathbf{I}$ will be studied in the future communication.

\subsection{Proposed Coordinate Minimization Approach}
\label{sec:proposed:method}

We propose to solve the following graph learning problem:
\begin{align}\label{eq:proposed_problem}
  \min_{\mathbf{w}\geq\mathbf{0}, \mathbf{q}>\mathbf{0}} &
    F\Bigl(\mathbf{B}\diag(\mathbf{w})\mathbf{B}^T, \diag(\mathbf{q})\Bigr)
  \text{.}
\end{align}
Note that the non-negativity of $\mathbf{w}$ and positivity of $\mathbf{q}$ are enough to ensure that $\mathbf{L}$ is Hermitian semi-definite positive and $\mathbf{Q}$ is Hermitian definite positive, such that the $(\mathbf{L}, \mathbf{Q})$-GFT is well-defined.
However, for our implementation, we need to introduce a hyperparameter $\qmin$ and change the positivity constraint to $\forall i,q_i\geq \qmin$.
This is justified below, in the update formula for $q_i$.

Coordinate minimization iterates through all edges and vertices and updates these weights according to \autoref{thm:coord_min}.

\begin{theorem}[Coordinate Minimization Update]\label{thm:coord_min}
  \autoref{eq:proposed_problem} is solved by, at iteration $t$, updating either $w_e$ or $q_i$ using:
  \begin{align*}
    \delta_e^{(t\mathrlap{)}} &= \max(\shortminus w_e^{(t\mathrlap{)}}, {\textstyle \frac{1}{h_e} \shortminus \frac{1}{r_e^{(t\mathrlap{)}}}}) &
    \delta_i^{(t\mathrlap{)}} &= \max(q_{\text{min}}\shortminus q_i^{(t\mathrlap{)}}, {\textstyle \frac{1}{p_i} \shortminus \frac{1}{u_i^{(t\mathrlap{)}}}})
  \end{align*}
  with edge cost $h_e=\mathbf{b_e}^{\smash{T}}\mathbf{S}\mathbf{b_e}$, effective resistance $r_e^{\smash{(t)}}=\mathbf{b_e}^{\smash{T}}(\mathbf{Q}^{\smash{(t)}}+\mathbf{L}^{\smash{(t)}})^{-1}\mathbf{b_e}$, vertex cost $p_i=\mathbf{S}_{ii}$, vertex effective importance $u_i^{\smash{(t)}}=\left[(\mathbf{Q}^{\smash{(t)}}+\mathbf{L}^{\smash{(t)}})^{-1}\right]_{ii}$, and $\qmin>0$.
\end{theorem}

\begin{proof}
  Notice first how the trace term in \autoref{eq:proposed_loglikelihood} can be decomposed as a sum $\trace(\mathbf{LS})+\trace(\mathbf{QS})$.
  Any update to $w_e$ (\resp $q_i$) will only modify the first (\resp second) trace in this sum.

  The derivation for edge weight update $\delta_e$ is identical to \cite{Pavez.ASILOMAR.2019} by changing $\frac{1}{N}\mathbf{J}$ to $\mathbf{Q}$ in the effective resistance $r_e$ definition.

  For the update $\delta_i$ to vertex importance $q_i$, we solve:
  \[
    \Delta F=-\log\left(1+\delta_i\left[(\mathbf{Q}+\mathbf{L})^{-1}\right]_{ii}\right)+\delta_i\mathbf{S}_{ii}
    \text{,}
  \]
  where the first term is derived from the log-determinant variation and the second from the trace variation.
  Let $p_i=\mathbf{S}_{ii}$, and $u_i=\left[(\mathbf{L}+\mathbf{Q})^{-1}\right]_{ii}$.
  Using the Lagrangian of the problem minimizing the variation of $F$ requires weak inequalities for the constraints.
  Therefore, we relax $q_i>0$ into $q_i\geq\qmin$ for some $\qmin>0$.
  The KKT conditions lead to:
  \[
    -\frac{u_i}{1+\delta_i u_i}+p_i-\lambda_i = 0, \lambda_i\geq 0,q^{(t+1)'}_i\geq 0, \lambda_iq^{(t+1)'}_i=0
    \text{,}
  \]
  with $q^{(t+1)'}_i = q^{(t+1)}_i-\qmin$.
  These conditions are satisfied in the proposed update.
\end{proof}

$u_i$ and $r_e$ depend on $\mathbf{L}$ or $\mathbf{Q}$ and need to be updated after each update.
We use the Sherman-Morrison formula to write the update to $\bm{\Phi^{(t)}}=\left(\mathbf{Q}+\mathbf{L}\right)^{-1}$ \cite{Pavez.ASILOMAR.2019}:
\begin{align*}
  \Delta\bm{\Phi^{(t)}} &=
    \left\{
    \begin{array}{ll}
      \frac
        {\delta_e^{(t)}\left(\bm{\Phi^{(t)}}\mathbf{b_e}\right)\left(\bm{\Phi^{(t)}}\mathbf{b_e}\right)^T}
        {1+\delta_e^{(t)}\mathbf{b_e}^T\bm{\Phi^{(t)}}\mathbf{b_e}} & \text{if }w_e\text{ is updated} \\
      \frac
        {\delta_i^{(t)}[\bm{\Phi^{(t)}}]_{.i}[\bm{\Phi^{(t)}}]_{.i}^T}
        {1+\delta_i^{(t)}[\bm{\Phi^{(t)}}]_{ii}} & \text{if }q_i\text{ is updated.} \\
    \end{array}
    \right.
\end{align*}
Using $\Delta\bm{\Phi^{(t)}}=\bm{\Phi^{(t+1)}}-\bm{\Phi^{(t)}}$, the updates to $r_f$ ($f\in\mathcal{E}$), and to $q_j$ ($j\in\mathcal{V}$) can then be easily obtained through:
\begin{align*}
  \Delta\mathbf{r_f^{(t)}} &= \mathbf{b_f}^T\Delta\bm{\Phi^{(t)}}\mathbf{b_f} &
  \Delta\mathbf{u_j^{(t)}} &= [\Delta\bm{\Phi^{(t)}}]_{jj}
\end{align*}
which is efficient to implement when $\bm{\Phi^{(t)}}=\left(\mathbf{Q}+\mathbf{L}\right)^{-1}$ is kept in memory and updated after each iteration.

\textbf{Connectedness}\quad
Compared to \cite{Pavez.ASILOMAR.2019}, edge weight updates can disconnect the graph since the matrix $\mathbf{Q}+\mathbf{L}$ is always non-singular, whereas $\mathbf{L}+\frac{1}{N}\mathbf{J}$ would become singular.
As a consequence, edge weight updates are not restricted, thus allowing to fully adapt to the data.

\textbf{Stopping Criterion}\quad
Coordinate minimization is stopped whenever the maximum number of epochs ($N+M$ updates) is reached, or improvement of the cost function $F$ is below a predefined threshold after an epoch.

\subsection{Optimal Graph Weight Properties}
\label{sec:proposed:edge_prop}

Similarly to the Generalized Graph Laplacian case of \cite{Pavez.AISTATS.2022}, we can upper bound edge weights, but only between vertices whose importances are larger than $q_{\text{min}}$.

\begin{theorem}\label{thm:upper_bound}
    Let $\rho_{ij}=\mathbf{S}_{ij}/\sqrt{\mathbf{S}_{ii}\mathbf{S}_{jj}}$ be the sample correlation coefficient between vertices $i$ and $j$.
    The optimal solution of problem \eqref{eq:proposed_problem} verifies for all non zero graph weights $w_{ij}>0$ with $q_i>q_{\text{min}}$ and $q_j>q_{\text{min}}$:
    \begin{equation}
        \mathbf{w}_{ij} \leq \frac{1}{|\mathbf{S}_{ij}|} \frac{\rho_{ij}^2}{1 - \rho_{ij}^2}
    \end{equation}
\end{theorem}

\begin{proof}
    Similar to \cite[Appendix A]{Pavez.AISTATS.2022}, but using the additional constraints $q_{\text{min}}-[\mathbf{\Theta}]_{ii}\leq0$.
\end{proof}

In addition, any negative correlation between vertices leads to no edge between them, thus allowing us to remove those edges before performing coordinate minimization:

\begin{theorem}\label{thm:edge_set}
    The optimal graph satisfies:
    \begin{equation}
        \mathcal{E} = \lbrace (i,j)\in\mathcal{V}^2 : w_{ij}>0 \rbrace \subset \lbrace  (i,j) : \mathbf{S}_{ij} > 0 \rbrace.
    \end{equation}
\end{theorem}

\begin{proof}
    Using the fact that $\mathbf{\Theta}$ is a generalized Laplacian ($M$-matrix), and the KKT conditions for the strict inequality.
\end{proof}

\section{Experiments}
\label{sec:experiments}

To experimentally validate our method, we use a synthetic stochastic signal whose statistics are defined by an underlying Euclidean space.
Our goal is to show that our approach can reliably extract meaningful properties of the Euclidean space from the learnt graph, without relying on where the observations are made in Euclidean space.

\begin{figure}[b]
    \centering
    \vspace{-1em}
    \begin{tikzpicture}[every node/.style={scale=0.8}]
        \begin{axis}[%
                axis x line=bottom,
                axis y line=left,
                xmin=0,xmax=1.5,
                ymin=0,ymax=10.4,
                xlabel=vertex distance $d$,ylabel=$2\gamma(d)$,
                ylabel near ticks,
                xlabel near ticks,
                ylabel shift=-5 pt,
                legend pos=south east,
                height=4cm,width=9cm]
            \addplot[thick,domain=0:1.5,samples=300] {10*(1-exp(-x/0.01))};
            \addlegendentry{$r=0.01$}
            \addplot[dashed,thick,domain=0:1.5,samples=300] {10*(1-exp(-x/0.02))};
            \addlegendentry{$r=0.02$}
            \addplot[dotted,thick,domain=0:1.5,samples=300] {10*(1-exp(-x/0.1))};
            \addlegendentry{$r=0.1$}
            \addplot[dashdotted,thick,domain=0:1.5,samples=300] {10*(1-exp(-x/0.2))};
            \addlegendentry{$r=0.2$}
            \addplot[loosely dashdotdotted,thick,domain=0:1.5,samples=300] {10*(1-exp(-x))};
            \addlegendentry{$r=1$}
        \end{axis}
    \end{tikzpicture}
    \vspace{-0.5em}
    \caption{Variograms being considered for our experiments.}
    \label{fig:variograms}
\end{figure}
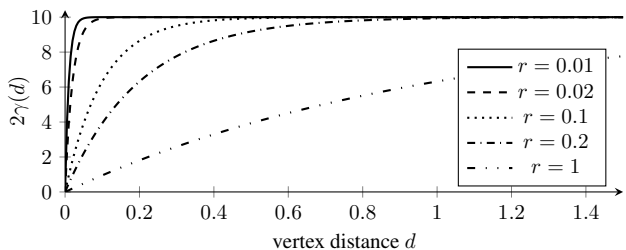

More precisely, we consider intrinsic stationary 2D signals \cite{Cressie.BOOK.1993}.
In this experiment, we use an iso\-tropic exponential variogram with no \textit{nugget} (\textit{i.e.} no measurement noise), a \textit{sill} (variance) of 10 for all data points, and varying values of the \textit{range} $r$:
\[
    2\gamma(d) = 10 \bigl(1 - \exp(-d/r)\bigr)
    \text{,}
\]
where $d$ is the Euclidean distance between any two points.
To observe the behavior of our proposed method, we choose 5 values of range (see \autoref{fig:variograms}).
Smaller ranges correspond to signals being correlated only for a short distance (almost white noise), while larger ranges being highly correlated throughout the space (almost constant signal).
A range of 0.1 corresponds to the interesting middle case where there is substantial correlation between somewhat close locations, but almost no correlation between locations far away.

We then uniformly sample the 2D space ($[0,1]\times[0,1]$ Euclidean plane) at $N=50$ locations.
To study the consistency of the learnt graph, we generate $K=50$ sets of random node locations.
For a given spatial sampling and range $r$, we obtain the exact covariance matrix $\mathbf{S}$ of the data with $\mathbf{S}_{ii}=10$ and $\mathbf{S}_{ij}=10\exp(-d_{ij}/r), \forall i\neq j$.
Using \autoref{thm:upper_bound}, we obtain the following edge weight upper bound (when $q_i>q_{\text{min}}<q_j$):
\begin{equation}
    w_{ij}\leq 0.1 \left(e^{d_{ij}/r} - e^{-d_{ij}/r}\right)^{-1}.
\end{equation}
\autoref{thm:edge_set} does not allow to remove any edge before performing coordinate minimization since $\mathbf{S}_{ij}>0$ for any two vertices $i$ and $j$.
The corresponding upper bound for \cite{Pavez.ASILOMAR.2019} is $w_{ij}\leq 0.05(1-\exp(-d_{ij}/r))^{-1}$.
These bounds are shown on \autoref{fig:bounds}.

\begin{figure}[b]
    \centering
    \vspace{-1em}
    {\protect\NoHyper
    \begin{tikzpicture}[every node/.style={scale=0.7}]
        \begin{axis}[%
                axis x line=bottom,
                axis y line=left,
                ymode=log,
                xmin=0,xmax=0.6,
                ymin=1e-6,ymax=50,
                xlabel=vertex distance $d$,
                xlabel near ticks,
                legend pos=south east,
                legend columns=5,legend transposed,
                height=5cm,width=9cm]
            \addplot[blue,ultra thick,domain=0:1.5,samples=500] {0.1/(exp(x/0.01)-exp(-x/0.01))};
            \addlegendentry{$r=0.01$ (proposed)}
            \addplot[dashed,blue,ultra thick,domain=0:1.5,samples=500] {0.1/(exp(x/0.02)-exp(-x/0.02))};
            \addlegendentry{$r=0.02$ (proposed)}
            \addplot[dotted,blue,ultra thick,domain=0:1.5,samples=500] {0.1/(exp(x/0.1)-exp(-x/0.1))};
            \addlegendentry{$r=0.1$ (proposed)}
            \addplot[dashdotted,blue,ultra thick,domain=0:1.5,samples=500] {0.1/(exp(x/0.2)-exp(-x/0.2))};
            \addlegendentry{$r=0.2$ (proposed)}
            \addplot[loosely dashdotdotted,blue,ultra thick,domain=0:1.5,samples=500] {0.1/(exp(x)-exp(-x))};
            \addlegendentry{$r=1$ (proposed)}
            \addplot[red,thin,domain=0:1.5,samples=300] {0.05/(1-exp(-x/0.01))};
            \addlegendentry{$r=0.01$ (\cite{Pavez.ASILOMAR.2019})}
            \addplot[dashed,red,thin,domain=0:1.5,samples=300] {0.05/(1-exp(-x/0.02))};
            \addlegendentry{$r=0.02$ (\cite{Pavez.ASILOMAR.2019})}
            \addplot[dotted,red,thin,domain=0:1.5,samples=300] {0.05/(1-exp(-x/0.1))};
            \addlegendentry{$r=0.1$ (\cite{Pavez.ASILOMAR.2019})}
            \addplot[dashdotted,red,thin,domain=0:1.5,samples=300] {0.05/(1-exp(-x/0.2))};
            \addlegendentry{$r=0.2$ (\cite{Pavez.ASILOMAR.2019})}
            \addplot[loosely dashdotdotted,red,thin,domain=0:1.5,samples=300] {0.05/(1-exp(-x))};
            \addlegendentry{$r=1$ (\cite{Pavez.ASILOMAR.2019})}
        \end{axis}
    \end{tikzpicture}
    \protect\endNoHyper}
    \vspace{-0.5em}
    \caption{Edge weight upper bounds for both \cite{Pavez.ASILOMAR.2019} and our proposed approach, and for varying ranges of the variogram.}
    \label{fig:bounds}
\end{figure}
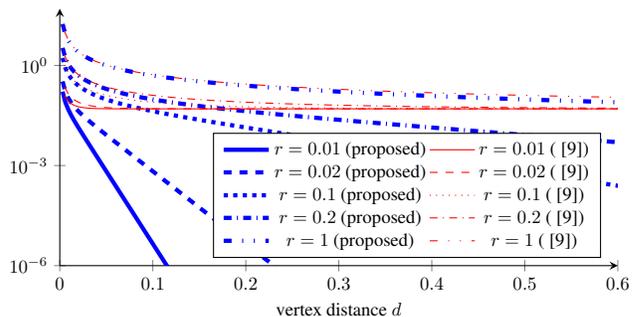

For each of these $K$ samplings, and each of the 5 ranges, we learn
\begin{inlinelist}
    \item a combinatorial Laplacian from the data covariance matrix $\mathbf{S}$  using \cite{Pavez.ASILOMAR.2019}
    \item an inner product matrix and a combinatorial Laplacian using our proposed approach
\end{inlinelist}.
In both cases, we use the same graph initialization (graph weights between vertices obtained from a Gaussian kernel of the distance between vertices, with $\sigma$ chosen as one third of the average distance) and let the algorithm run until the cost function is not changed by more than $10^{-10}$.
Averages for key metrics are shown in \autoref{tab:metrics}.

We first study how vertex importance $\mathbf{q}$ changes with range $r$.
We are interested in two key quantities: how many of the vertices get the minimum importance $q_{\text{min}}$? and what is the average importance of the remaining vertices?
We define the first quantity as $u(\mathbf{q})=|\{q_i:q_i=q_{\text{min}}\}|/N$ ($u$ for "unimportant").
\autoref{tab:metrics} shows that all vertices are important for smaller ranges, and fewer and fewer vertices are important as range increases.
Intuitively, for lower ranges, correlation between neighbors fades very quickly, as shown on \autoref{fig:variograms}, therefore, neighbors are not enough to explain the value of a vertex and all vertices are important to model the signal.
However, for larger ranges, the value on a given vertex can be accurately inferred from its neighbors values, due to the high correlation: only a small fraction of vertices are important to accurately model the whole signal.

For the vertices that are important ($q_i>q_{\text{min}}$), we also observe~\cite{SupportingMaterial} that importance is higher for smaller ranges, while decreasing with range.
The average importance $\bar{q}=\langle\{q_i:q_i>q_{\text{min}}\}\rangle$ in \autoref{tab:metrics} shows the relation between importance and range.
This follows from the cost function \autoref{eq:proposed_loglikelihood}, where compared to the cost \autoref{eq:asilomar_cost} of \cite{Pavez.ASILOMAR.2019}, $\mathbf{Q}$ captures the information of $\mathbf{S}$ not captured by $\mathbf{L}$.
For larger ranges, $\mathbf{L}$ is a good model using \cite{Pavez.ASILOMAR.2019}, and vertex importances remain low as they are not needed, while for smaller ranges, vertex importances are larger because $\mathbf{L}$ is not enough for a good model.

These observations are especially important in the context of sampling on graphs \cite{Girault.ICASSP.2020}, where the goal is to sample so as to minimize the $\mathbf{Q}$-norm of the error.
With our proposed approach, the resulting importances target specific ver\-tices whose values are important to keep to reconstruct accurately the graph signal.
%

We expect from the edge weight bounds in  \autoref{fig:bounds} that the graphs learnt with~\cite{Pavez.ASILOMAR.2019} should be denser when the signal has lower correlation, which corresponds here to lower ranges.
We define sparsity as the proportion of edges having 0 weight leading to: $\epsilon(\mathbf{w})=\|\mathbf{w}\|_0/\smash{\frac{N(N-1)}{2}}$.
This metric confirms that graphs are sparser with our proposed approach (see \autoref{tab:metrics}).

Our supplementary materials also show that edge weights are generally much smaller for the smaller ranges of 0.01 and 0.02~\cite{SupportingMaterial}, by almost 4 orders of magnitude.
We also remark that many of those edge weights violate the upper bound in \autoref{thm:upper_bound}, and are thus interpreted as numerical errors in the resolution.
Trimming edges based on this could actually lead to even sparser graphs without sacrificing precision.
We will study this in a future communication.

\begin{table}[t]
    \caption{
        Average key properties of the learnt graph between \cite{Pavez.ASILOMAR.2019} and our proposed approach.
        Bold values are better.
    }
    \vspace{-0.5em}
    \label{tab:metrics}
    \centering
    \footnotesize
    \begin{tabular}{ccccccr}
        \toprule
        Method & $r$ & $u(\mathbf{q})$ & $\bar{q}$ & $\epsilon(\mathbf{w})$ & Time \\
        \midrule
        \multirow{4}{*}{\cite{Pavez.ASILOMAR.2019}}
            & 0.01 && & 0\% & 299.2s \\
            & 0.02 && & 3.3e-3\% & 292.9s \\
            & 0.1 && & 10.6\% & 295.9s \\
            & 0.2 && & 71.0\% & 177.4s \\
            & 1 && & 90.5\% & \textbf{40.8s} \\
        \midrule
        \multirow{4}{*}{Proposed}
            & 0.01 & 0\% & 9.7e-2 & \textbf{51.8\%} & \textbf{58.7s} \\
            & 0.02 & 0\% & 9.1e-2 & \textbf{66.5\%} & \textbf{76.1s} \\
            & 0.1 & 0.3\% & 3.2e-2 & \textbf{85.1\%} & \textbf{71.2s} \\
            & 0.2 & 11\% & 1.5e-2 & \textbf{88.6\%} & \textbf{59.8s} \\
            & 1 & 75\% & 1.3e-2 & \textbf{90.7\%} & 53.2s \\
        \bottomrule
    \end{tabular}
    \vspace{-1.5em}
\end{table}

\section{Conclusions and Perspectives}

In this paper, we proposed a graph learning approach based on a different graph signal spectral model, and based on learning jointly vertex importances and edge weights.
As shown experimentally with continuous intrinsically  stationary signals, such an approach allows for dramatically sparser graphs in the low correlation regime, more interpretable weights, and vertex importances highlighting which vertices are enough to model the signal.
Adding the freedom of choosing vertex importances, this effectively lessens overfitting through a richer, more accurate, graph signal model space.
Future work will include proposing edge screening similar to \cite{Pavez.ASILOMAR.2019}, regularization, evaluation with more adverse synthetic data and real data, and theoretical guarantees.

\vspace{-1em}
\section{References}

{\def\section#1{}
\footnotesize
\bibliographystyle{IEEEbib}
\bibliography{bibliography}

\begin{thebibliography}{10}

\bibitem{Ortega.PROCIEEE}
A.~{Ortega}, P.~{Frossard}, J.~{Kova{\v c}evi{\'c}}, J.~M.~F. {Moura}, and
  P.~{Vandergheynst},
\newblock ``{Graph Signal Processing: Overview, Challenges, and
  Applications},''
\newblock {\em Proc. of the IEEE}, vol. 106, no. 5, pp. 808--828, May 2018.

\bibitem{Ortega.BOOK.2022}
A.~Ortega,
\newblock {\em Introduction to Graph Signal Processing},
\newblock Cam Uni P, 2022.

\bibitem{dong2019learning}
X.~Dong, D.~Thanou, M.~Rabbat, and P.~Frossard,
\newblock ``Learning graphs from data: A signal representation perspective,''
\newblock {\em IEEE Signal Processing Magazine}, vol. 36, no. 3, pp. 44--63,
  2019.

\bibitem{Mateos.SPMAG.2019}
G.~Mateos, S.~Segarra, A.~G. Marques, and A.~Ribeiro,
\newblock ``{Connecting the Dots: Identifying Network Structure via Graph
  Signal Processing},''
\newblock {\em IEEE Signal Processing Magazine}, vol. 36, no. 3, pp. 16--43,
  May 2019.

\bibitem{Egilmez.JSTSP.2017}
H.~E. Egilmez, E.~Pavez, and A.~Ortega,
\newblock ``{Graph Learning From Data Under Laplacian and Structural
  Constraints},''
\newblock {\em {IEEE Journal of Selected Topics in Signal Processing}}, vol.
  11, no. 6, pp. 825--841, Sept. 2017.

\bibitem{Girault.TSP.2018}
B.~Girault, A.~Ortega, and S.~S. Narayanan,
\newblock ``{Irregularity-Aware Graph Fourier Transforms},''
\newblock {\em {IEEE Transactions on Signal Processing}}, vol. 66, no. 21, pp.
  5746--5761, Nov. 2018.

\bibitem{Girault.ICASSP.2014}
B.~Girault, P.~Gon{\c c}alves, {\'E}.~Fleury, and A.~S. Mor,
\newblock ``{Semi-supervised learning for graph to signal mapping: A graph
  signal wiener filter interpretation},''
\newblock in {\em {2014 IEEE International Conference on Acoustics, Speech and
  Signal Processing (ICASSP)}}, May 2014.

\bibitem{Marques.TSP.2017}
A.~G. {Marques}, S.~{Segarra}, G.~{Leus}, and A.~{Ribeiro},
\newblock ``{Stationary Graph Processes and Spectral Estimation},''
\newblock {\em IEEE Transactions on Signal Processing}, vol. 65, no. 22, pp.
  5911--5926, Nov 2017.

\bibitem{Pavez.ASILOMAR.2019}
E.~Pavez and A.~Ortega,
\newblock ``{An Efficient Algorithm for Graph Laplacian Optimization Based on
  Effective Resistances},''
\newblock in {\em {2019 53rd Asilomar Conf. on Signals, Systems, and
  Computers}}, Nov. 2019, pp. 51--55.

\bibitem{Pavez.AISTATS.2022}
E.~Pavez,
\newblock ``Laplacian constrained precision matrix estimation: Existence and
  high dimensional consistency,''
\newblock in {\em Proc. of the 25th International Conference on Artificial
  Intelligence and Statistics}, Mar 2022.

\bibitem{ying2020nonconvex}
J.~Ying, J.~V. de~M. Cardoso, and D.~P. Palomar,
\newblock ``Nonconvex sparse graph learning under laplacian constrained
  graphical model,''
\newblock {\em Advances in Neural Information Processing Systems}, vol. 33, pp.
  7101--7113, 2020.

\bibitem{ying2021minimax}
J.~Ying, J.~V. de~M. Cardoso, and D.~P. Palomar,
\newblock ``Minimax estimation of laplacian constrained precision matrices,''
\newblock in {\em International Conference on Artificial Intelligence and
  Statistics}. PMLR, 2021, pp. 3736--3744.

\bibitem{koyakumaru2023learning}
T.~Koyakumaru, M.~Yukawa, E.~Pavez, and A.~Ortega,
\newblock ``Learning sparse graph with minimax concave penalty under gaussian
  markov random fields,''
\newblock {\em IEICE Trans. on Fund. of Elec., Comm. and Comp.r Sc.}, 2023.

\bibitem{Pavez.ICASSP.2016}
E.~Pavez and A.~Ortega,
\newblock ``{Generalized Laplacian precision matrix estimation for graph signal
  processing},''
\newblock in {\em {2016 IEEE International Conference on Acoustics, Speech and
  Signal Processing (ICASSP)}}. Mar. 2016, pp. 6350--6354, IEEE.

\bibitem{slawski2015estimation}
M.~Slawski and M.~Hein,
\newblock ``Estimation of positive definite m-matrices and structure learning
  for attractive gaussian markov random fields,''
\newblock {\em Linear Algebra and its Applications}, vol. 473, pp. 145--179,
  2015.

\bibitem{ying2022adaptive}
J.~Ying, J.~V. de~M. Cardoso, and D.~P. Palomar,
\newblock ``Adaptive estimation of mtp2 graphical models,''
\newblock {\em arXiv preprint arXiv:2210.15471}, 2022.

\bibitem{Pavez.TSP.2018}
E.~Pavez, H.~E. Egilmez, and A.~Ortega,
\newblock ``{Learning Graphs With Monotone Topology Properties and Multiple
  Connected Components},''
\newblock {\em {IEEE Trans. on Signal Processing}}, vol. 66, no. 9, pp.
  2399--2413, May 2018.

\bibitem{Girault.ICASSP.2020}
B.~Girault, A.~Ortega, and S.~S. Narayayan,
\newblock ``{Graph Vertex Sampling with Arbitrary Graph Signal Hilbert
  Spaces},''
\newblock in {\em {2020 IEEE Int. Conf. on Acoustics, Speech and Sig. Proc.
  (ICASSP)}}, May 2020.

\bibitem{Lu.ICIP.2020}
K.-S. Lu, A.~Ortega, D.~Mukherjee, and Y.~Chen,
\newblock ``{Perceptually Inspired Weighted MSE Optimization Using
  Irregularity-Aware Graph Fourier Transform},''
\newblock in {\em {2020 IEEE Int. Conf. on Im. Proc.}}, 2020.

\bibitem{Pavez.ICIP.2020}
E.~Pavez, B.~Girault, A.~Ortega, and P.~A. Chou,
\newblock ``{Region Adaptive Graph Fourier Transform for 3D Point Clouds},''
\newblock in {\em {2020 IEEE Int. Conf. on Im. Proc.}}, Oct. 2020.

\bibitem{Pavez.ARXIV.2022}
E.~Pavez, B.~Girault, A.~Ortega, and P.~A. Chou,
\newblock ``{Two Channel Filter Banks on Arbitrary Graphs with Positive Semi
  Definite Variation Operators},'' Feb. 2023,
\newblock arXiv: 2203.02858.

\bibitem{Girault.EUSIPCO.2015}
B.~Girault,
\newblock ``{Stationary Graph Signals using an Isometric Graph Translation},''
\newblock in {\em {2015 23rd Europ. Sig. Proc. Conf. (EUSIPCO)}}, 2015.

\bibitem{lake2010discovering}
B.~Lake and J.~Tenenbaum,
\newblock ``Discovering structure by learning sparse graphs,''
\newblock in {\em Proc. of the Ann. Meet. of the Cog. Sc. Soc.}, 2010.

\bibitem{Cressie.BOOK.1993}
N.~A.~C. Cressie,
\newblock {\em {Statistics for spatial data}},
\newblock Wiley, Jan. 1993.

\bibitem{SupportingMaterial}
``Supp.,'' https://www.benjamin-girault.com/pages/graph\_learning.html.

\end{thebibliography}
}

\ifnum\eusipco=0

\clearpage

\section{Extended Proofs}

\begin{proof}[Proof of \autoref{thm:upper_bound}]
    This proof follows the approach of \cite[Appendix A]{Pavez.AISTATS.2022}, but using an additional constraint enforcing positive $q_i$.
    Let $\mathbf{\Theta}=\mathbf{Q}+\mathbf{L}$ be the precision matrix we are estimating, such that $\mathbf{\Theta}_{ii}=q_i+d_i$ and $\mathbf{\Theta}_{ij}=-w_{ij}$.
    After relaxing the constraint $q_i>0$, Problem \eqref{eq:proposed_problem} is then equivalent to:
    \begin{align*}
        \min_{\substack{
            \mathbf{\Theta}_{ij}\leq 0, i\neq j \\
            q_{\text{min}}\mathbf{1}-\mathbf{\Theta}\mathbf{1} < \mathbf{0}
        }}
        &
        -\logdet(\mathbf{\Theta})
        +\trace(\mathbf{\Theta}\mathbf{S})
        \text{,}
    \end{align*}
    where we use the property that $[\mathbf{\Theta}\mathbf{1}]_i=q_i$.
    Its Lagrangian uses the symmetric Lagrange multipliers $\lambda_{ij}$ for the constraints on $\mathbf{\Theta}_{ij}$ and $\mu_i$ for the constraint on the row sums of $\mathbf{\Theta}$, leading to the following KKT conditions:
    \begin{align*}
        -\mathbf{\Theta}^{-1} + \mathbf{S} + \mathbf{\Lambda} + \mathbf{M} &= \mathbf{0} & \mathbf{\Theta} &\succcurlyeq 0 \\
        \lambda_{ij}\mathbf{\Theta}_{ij} &= 0, \forall i\neq j & \mu_i\left(q_{\text{min}} - [\mathbf{\Theta}\mathbf{1}]_i\right) &= 0, \forall i \\
        \mathbf{\Theta}_{ij} &\leq 0, \forall i\neq j & q_{\text{min}} - [\mathbf{\Theta}\mathbf{1}]_i &\leq 0, \forall i \\
        \lambda_{ij} &\geq 0, \forall i\neq j & \mu_i &\geq 0, \forall i
    \end{align*}
    where the Lagrange multiplier matrices verify $\mathbf{\Lambda}_{ii}=0$, $\mathbf{\Lambda}_{ij}=\mathbf{\Lambda}_{ji}=\lambda_{ij}$, $\mathbf{M}_{ii}=-\mu_i$, and $\mathbf{M}_{ij}=-\mu_i-\mu_j$.
    Considering the subset of vertices $\mathcal{S}=\{i,j\}$, and using the same technique as \cite{Pavez.AISTATS.2022}, the first condition (gradient is 0) leads to:
    \begin{equation}\label{eq:theta_ij_upper_bound}
        \mathbf{\Theta}_{\mathcal{S}, \mathcal{S}} \geq \left(\mathbf{S}_{\mathcal{S}, \mathcal{S}} + \mathbf{\Lambda}_{\mathcal{S}, \mathcal{S}} + \mathbf{M}_{\mathcal{S}, \mathcal{S}}\right)^{-1}
        \text{.}
    \end{equation}
    Inverting the $2\times2$ matrix on the r.h.s, and considering an edge $ij$ such that its optimal weight $w_{ij}>0$ leads to:
    \begin{align*}
        \mathbf{\Theta}_{ij}
            &\geq \frac{-(\mathbf{S}_{ij} + \lambda_{ij} - \mu_i-\mu_j)}%
                      {(\mathbf{S}_{ii} - \mu_i)(\mathbf{S}_{jj} - \mu_j) - (\mathbf{S}_{ij} + \lambda_{ij} - \mu_i - \mu_j)^2} \\
            &\geq -\left[\frac{(\mathbf{S}_{ii} - \mu_i)(\mathbf{S}_{jj} - \mu_j)}{\mathbf{S}_{ij} - \mu_i - \mu_j} - (\mathbf{S}_{ij} - \mu_i - \mu_j)\right]^{-1} \\
            &\geq -\left[\frac{\mathbf{S}_{ii}\mathbf{S}_{jj}}{\mathbf{S}_{ij}} - \mathbf{S}_{ij}\right]^{-1}
    \end{align*}
    where the second inequality is obtained using $\lambda_{ij}=0$ since $w_{ij}>0$, and the last using $\mu_i=0=\mu_j$ since $q_i>q_{\text{min}}$ and $q_j>q_{\text{min}}$.
\end{proof}

\begin{proof}[Proof of \autoref{thm:edge_set}]
    To prove the result, we prove that for non-zero edge weight $w_{ij}>0$ of the optimal graph we have $\mu_i+\mu_j<\mathbf{S}_{ij}$.
    The result will then follow by observing that $\mu_i\geq 0$ and $\mu_j\geq 0$.
    Let $i,j$ such that $w_{ij}>0$.
    From the first equation in the KKT conditions we have that:
    \begin{equation}
        0 \leq (\mathbf{\Theta}^{-1})_{ij} = \mathbf{S}_{ij} + \lambda_{ij} - \mu_i - \mu_j.
    \end{equation}
    The inequality comes because $\mathbf{\Theta}$ is a generalized Laplacian ($M$-matrix).
    Since $\lambda_{ij}=0$, then $\mathbf{S}_{ij}  - \mu_i - \mu_j \geq 0$.
    For the strict inequality, we further assume that $\mu_i + \mu_j = \mathbf{S}_{ij}$ and show a contradiction.
    Using \eqref{eq:theta_ij_upper_bound}, we obtain:
    \[
        \mathbf{\Theta}_{ij} \geq \frac{-\lambda_{ij}}{(\mathbf{S}_{ii} - \mu_i)(\mathbf{S}_{jj} - \mu_j) - \lambda_{ij}^2}
        \text{.}
    \]
    Since $\lambda_{ij}=0$, then $0>\shortminus w_{ij}=\mathbf{\Theta}_{ij}\geq 0$, which is impossible.
\end{proof}

\fi

\end{document}